\documentclass[11pt]{article}
\usepackage{fullpage, amsmath, amssymb, amsthm, graphicx}

\title{Thresholds for Extreme Orientability}

\author{
Po-Shen Loh\thanks{Department of Mathematical Sciences, Carnegie Mellon
  University, Pittsburgh, PA 15213. Email: ploh@cmu.edu. Research
  supported by an NSA Young Investigators Grant and a USA-Israel BSF Grant.}  
\and
Rasmus Pagh\thanks{IT University of Copenhagen, Copenhagen, Denmark. Email: pagh@itu.dk.}
}

\date{}

\newtheorem{theorem}{Theorem}[section]
\newtheorem{definition}[theorem]{Definition}
\newtheorem{lemma}[theorem]{Lemma}

\newcommand{\pr}[1]{\mathbb{P}\left[#1\right]}
\newcommand{\E}[1]{\mathbb{E}\left[#1\right]}
\newcommand{\bin}[1]{\text{Bin}\left[#1\right]}

\begin{document}
\maketitle

\begin{abstract}
	Multiple-choice load balancing has been a topic of intense study since the seminal paper of Azar, Broder, Karlin, and Upfal.
	Questions in this area can be phrased in terms of {\em
        orientations\/} of a graph, or more generally a $k$-uniform random hypergraph.
	A {\em $(d,b)$-orientation\/} is an assignment of each edge to $d$ of its vertices, such that no vertex has more than $b$ edges assigned to it.
	Conditions for the existence of such orientations have been
        completely documented except for the ``extreme'' case of $(k-1,1)$-orientations.
	We consider this remaining case, and establish:
	\begin{itemize}
	\item The density threshold below which an orientation exists with high probability, and above which it does not exist with high probability.
	\item An algorithm for finding an orientation that runs in linear
          time with high probability, with explicit polynomial bounds on the
          failure probability.
	\end{itemize}
	Previously, the only known algorithms for constructing $(k-1,1)$-orientations worked for $k\leq 3$, and were only shown to have {\em expected\/} linear running time.  
\end{abstract}

\medskip


\medskip

\section{Introduction}

The efficiency of many algorithms and data structures rests on the fact
that randomly and independently throwing $m$ balls into $n$ bins ensures a
distribution that is, with high probability, close to uniform.  Since the
seminal paper of Azar et al.~\cite{Azar:1999:BA} a large literature has
grown around even stronger {\em multiple-choice\/} load balancing schemes
where the location of each ball is selected within a random {\em set\/} of
$k>1$ bins.

These problems have been studied both in the {\em on-line\/} setting, where
balls and their possible locations are revealed one by one, and in the {\em
off-line\/} setting where we are interested in the best allocation of a
given set of balls.  Most often, the focus of multiple-choice schemes is on
minimizing the maximum number of balls contained in any bin.  The question
can also be turned around to ask for the largest number of balls that can
be placed such that there are at most $b$ balls in each bin.  Of course,
this number depends on the random choices made, but in the off-line setting
it turns out that there is a well-defined {\em threshold\/} $m=(1\pm
o(1))\alpha n$, below which it is highly likely that the allocation is
possible, and above which it is highly unlikely that the allocation is possible.
Here, $\alpha$ is a constant that depends on $k$ and $b$, but not on $n$.

In this paper we consider the scenario where each ball comes in {\em $d$
copies\/}, and must be placed in exactly $d$ (distinct) out of $k$
possible bins.  Observe that the case $d=k$ is not so interesting, because
it is equivalent to the single-choice case with $md$ balls.  Thus the
interesting extreme case is $d = k-1$, which is the focus of this paper.
Motivation for copying each ball comes from parallel and distributed
systems where we  want high redundancy (resistance to $d-1$ failures),
and/or want to ensure that any set of balls can be accessed in parallel
with only a single request per bin.  Early papers investigating such
schemes
include~\cite{DietzfelbingerMey93,JACM::UpfalW1987,SICOMP::StockmeyerV1984}.
As a more recent example, Amossen and Pagh~\cite{AP11} considered the case
$k=3$, $d=2$, $b=1$, and showed that up to $\frac{(1-\varepsilon)}{6}m$
balls can be placed with high probability,\footnote{Meaning probability
tending to $1$ as $n\rightarrow\infty$.} for any constant $\varepsilon >
0$.  This was used to construct a data structure for sets that allows very
fast computation of set intersections on graphics hardware.  In this paper
we show that the constant 6 in~\cite{AP11} cannot be reduced, i.e., that
$m=(1\pm o(1)) n/6$ is the threshold for the problem of allocating balls
into $2$ of $3$ bins with maximum load~$1$.  In fact, we generalize this
result to the extreme case $d=k-1$, $b=1$ for any $k>2$, giving explicit
bounds on the probability of successful allocation in terms of $m$.  We
also present a generalization of the algorithms of~\cite{AP11,cuckoo-jour}
that  computes a $(k-1,1)$-orientation (if one exists) of a given
$k$-regular hypergraph, and show that it runs in linear time with high
probability.  This strengthens~\cite{AP11} which only shows linear running
time in expectation.

\subsection{Related work}

Multiple-choice balls and bins scenarios can be modeled as a random
$k$-regular hypergraph with $m$ edges (balls) on $n$ vertices (bins), where
edges are chosen i.i.d.\ uniformly from the set of all $k$-sets of
vertices.  Let $H_{n,m;k}$ be the random $k$-uniform hypergraph with $n$
vertices and $m$ hyperedges, where each such object is taken with equal
probability.  In the regime of interest in this work, when $m$ is linear in
$n$, there is essentially no difference between allowing and disallowing
multiple edges, because for $k \geq 3$, the probability that the
multi-hypergraph analogue repeats an edge is only $O(n^{-1})$.  Given such
a hypergraph, a {\em $(d,b)$-orientation\/} is an assignment of each edge
to $d$ of its vertices, such that no vertex has more than $b$ edges
assigned to it.

\paragraph{On-line setting.}
In our description of the on-line setting, we restrict attention to the case
where balls cannot be moved, once placed into bins.  Azar et
al.~\cite{Azar:1999:BA} considered $(1,b)$-orientations in the on-line
setting, and showed that the greedy algorithm that always assigns a ball to
its least loaded bin achieves a $(1,O(m/n+\log\log m /\log
k))$-orientation.  Tighter bounds for the maximum load of
$(1,b)$-orientations were later obtained by Berenbrink et
al.~\cite{Berenbrink:2000:BAH}.

\paragraph{Off-line setting.}
In the off-line setting, the threshold for $(1,1)$-orientations with $k=2$
can be shown (see, e.g.~\cite{cuckoo-jour}) to coincide with the appearance
of a {\em giant component\/} in the random graph, which is known to happen
at $m=(1\pm o(1)) n/2$ with high probability~\cite{ER-random-graphs}.
Several groups of
researchers~\cite{xorsat,conf/icalp/FountoulakisP10,Frieze:2009}
independently established the thresholds for $(1,1)$-orientations for every
$k>2$.  Generalizing in another direction, Fernholz and
Ramachandran~\cite{FR07} and Cain, Sanders, and Wormald~\cite{CSW07} showed
thresholds for $(1,b)$-orientations for $k=2$, and gave expected linear
time algorithms for computing an orientation.  This result was later
extended to $k>2$ by Fountoulakis et
al.~\cite{DBLP:conf/soda/FountoulakisKP11}.
Gao and Wormald~\cite{DBLP:conf/stoc/GaoW10} established thresholds for $(d,b)$-orientations, given that $b$ is a sufficiently large constant (depending on $d$ and $k$).
Independently of our work, Lelarge \cite{lelarge} recently developed new technical
machinery for this problem, which handles all combinations of the
parameters
$k$, $d$, and $b$ that satisfy $\max(k-d,b)\geq 2$.

\subsection{Our contribution}

In this paper we consider the remaining ``extreme'' case of
$\max(k-d,b)=1$, i.e., $d = k-1$ and $b=1$.
For this, we highlight two links between
Probabilistic Combinatorics and $(k-1,1)$-orientations.  First, we observe 
the connection between the literature on the phase transition in random
hypergraphs and $(k-1,1)$-orientations, which provides a natural
explanation for the threshold phenomenon experimentally documented
in~\cite{AP11}. Second,  we derive explicit, quantitative high-probability
bounds for the subcritical running time, by tracking a key parameter known
as ``susceptibility,'' through the Differential Equations method for
analyzing discrete random processes.  Previous bounds were only of
expected-time type.  Also, since we seek good polynomial-type dependencies
in our probability bounds, we perform a more careful analysis of the
susceptibility growth, which is substantially sharper than in previous
published work (e.g., \cite{BFKLS}) which was satisfied with error bounds
that could tend to zero very slowly. Our main theorem refers to the pseudocode of the {\sc Orient} algorithm, which can be found in section~\ref{sec:algorithm}. Its running time is determined by the number of iterations, which we define to be the number of times the condition in the {\bf while} loop is evaluated.

\begin{theorem}
  Let $0 < \epsilon < \frac{1}{2}$ be given, and assume that
  $\frac{n}{\log^6 n} > \frac{40000 k^6}{\epsilon^{12}}$.  Let $m =
  (1-\epsilon) \frac{n}{k(k-1)}$.  With probability at least $1 - 3
  n^{-1}$, all edges of the random $k$-uniform hypergraph $H_{n,m;k}$ can be 
  $(k-1, 1)$-oriented by the {\sc Orient} procedure using a total of at most
  \begin{displaymath}
    2 k^2 \left(
    \frac{1}{\epsilon} 
    + \frac{ 200 k^3 \log^3 n }{\epsilon^7 \sqrt{n}} 
    \right) \cdot n
    \,.
  \end{displaymath}
  iterations, each taking constant time.
  \label{thm:main}
\end{theorem}

This paper is organized as follows.  The next section observes the
natural threshold for extreme orientability.  Then, Section 3 applies the
Differential Equations method to deduce quantitative high-probability
bounds for algorithmic performance in the feasible regime.  The following
(standard) asymptotic notation will be utilized extensively.  For two
functions $f(n)$ and $g(n)$, we write $f(n) = o(g(n))$ or $g(n) =
\omega(f(n))$ if $\lim_{n \rightarrow \infty} f(n)/g(n) = 0$, and $f(n) =
O(g(n))$ or $g(n) = \Omega(f(n))$ if there exists a constant $M$ such that
$|f(n)| \leq M|g(n)|$ for all sufficiently large $n$.  


\section{Non-orientability}

We now investigate why there is no $(k-1,k)$-orientation when the number of edges exceeds $\frac{n}{k(k-1)}$. 
This is done by exhibiting an
obstruction that appears asymptotically almost surely as $n$
approaches infinity.  
One may observe many types of possible obstructions
to orientability.  A simple example for $k > 3$ is the $k$-uniform
hypergraph consisting of two hyperedges overlapping in three vertices.
It is clearly impossible to pick $k-1$ vertices for each hyperedge, as
there are only $2k-3$ vertices to share.  Unfortunately, any
fixed-size obstruction has a threshold for appearance in $H_{n,m;k}$
that is far beyond $\frac{n}{k(k-1)}$, so one cannot simply pinpoint
a single such hypergraph as the culprit for non-orientability.

Instead, we draw inspiration from the case $k=2$ (often referred to as ``cuckoo hashing''~\cite{cuckoo-jour}) where the desired
threshold $\frac{n}{2}$ matches the appearance of the well-studied
\emph{giant component}.  Indeed, the seminal result of Erd\H{o}s and
R\'enyi~\cite{ER-random-graphs} established that in the uniformly
random graph with $cn$ edges, for constants $c < \frac{1}{2}$, the
largest connected component has size $O(\log n)$, whereas for
constants $c > \frac{1}{2}$, the largest connected component has size
$\Omega(n)$.  Further study (see, e.g., the book \cite{JLR-book})
revealed that for $c < \frac{1}{2}$, all connected components are
either trees or unicyclic (containing at most one cycle), whereas for
$c > \frac{1}{2}$, the giant component is multicyclic.  As any
multicyclic component would have too many edges for vertices to be
$(k-1,1)$-orientable, this would establish the result for $k=2$.

The remainder of this section translates the random graph literature into
the orientability context, to observe the threshold for $k \geq 3$.  First,
it is convenient to introduce a measure of how ``crowded'' a component is.

\begin{definition}
  Let $k \geq 3$ be a fixed integer, and let $H$ be a $k$-uniform
  hypergraph.  The \textbf{excess}\/ of~$H$ is the difference $(k-1)
  e(H) - v(H)$, where $v(H)$ and $e(H)$ denote the numbers of vertices
  and edges in $H$, respectively.
\end{definition}

For connected hypergraphs $H$, the excess is always an integer greater
than or equal to $-1$.  When it is $-1$, the hypergraph is acyclic, and called a
\emph{hypertree}.  When the excess is 0, we say that $H$ is
\emph{unicyclic}, and when the excess is positive, we say that $H$ is
\emph{complex}.  Note that in the context of $(k-1,1)$-orientability, 
any complex component is an obstruction.  
Given an edge set $E'$ we define its {\em capacity} as $\text{cap}(E') =
\sum_{v\in V} \min (b,|\{e\in E' : v\in e\}|)$.
We have the following consequence of the max-flow min-cut theorem (see, e.g.~\cite[Section 6.1]{ref:Papadimitriou:1998a}):
\begin{theorem}
A $k$-regular hypergraph $(V,E)$ has a $(d,b)$-orientation if and only if each subset $E'\subseteq E$ has capacity $\text{cap}(E')\geq |E'|d$.
\end{theorem}
\begin{proof}
The capacity sums, over each vertex, an upper bound on how many edges in $E'$ can be oriented towards it.
If some edge set $E'$ has capacity less than $|E'|d$ it is thus impossible to orient all its edges (even ignoring edges outside of $E'$).
For the reverse direction consider the flow network with:
\begin{itemize}
\item Node set $E \cup V \cup \{s,t\}$, i.e., a node per edge and vertex in $(V,E)$, plus a source node $s$, and a sink node $t$.
\item Capacity 1 edges connecting the node of each $e\in E$ to the $k$ nodes in $V$ contained in $e$.
\item Capacity $d$ edges from $s$ to each vertex in $E$, and capacity $b$ edges from vertex in $V$ to $t$. 
\end{itemize}
Observe that an integer $s$-$t$ flow corresponds to an orientation of edges
with a flow of 1 from an edge to each vertex that the edge is oriented
towards. This means that if there is no $(d,b)$-orientation, there is no
integer $s$-$t$ flow of value $|E|d$. Since all capacities in the network
are integral, this in turn means that there exists no flow of value $|E|d$
at all. Using the max-flow min-cut theorem this implies that there is a
minimum $s$-$t$ cut $(S,T)$ such that the total capacity of edges from $S$
to $T$ is $\text{cut}(S,T)<|E|d$. Let $E'$ denote the set of edges that are
members of $S$. Since $(S,T)$ is minimal vertices in $V\cap S$ appear in at
least $b$ edges in $E'$, and vertices in $V\cap T$ appear in at most $b$
edges of $E'$. Thus we obtain:
\begin{align*}
\text{cap}(E') & = \sum_{v\in V\cap S} b + \sum_{v\in V\cap T} |\{e\in E' : v\in e\}|)
 = \text{cut}(S,T) - |E\backslash E'|d
 < |E'| d \enspace .
\end{align*}

\end{proof}
For $b=1$ the capacity of a set $E'$ is exactly the number of distinct vertices in its edges, so the capacity of $E'$ is $(k-1)|E'|$ minus the excess of $E'$. This means that $(k-1,1)$-orientability exactly coincides with the appearance of a complex component.
Much is known about the phase transition in random hypergraphs.  The
following results are from the paper \cite{KL-hypergraph-excess} of
Karo\'nski and \L uczak, which actually determines several results of
much higher precision.

\begin{theorem}
  (Theorem 4 in \cite{KL-hypergraph-excess}.)  Let $k \geq 3$ be a
  fixed integer, and let $m = \frac{n}{k(k-1)} - t(n)$, where $t(n)$
  is any function of higher order than $n^{2/3}$, i.e., $t(n) =
  \omega(n^{2/3})$.  Then $H_{n,m;k}$ consists of hypertrees
  and unicyclic components with high probability (as $n$ grows).
\end{theorem}

\noindent \textbf{Remark.} The previous argument establishes that hypertrees and
unicyclic components can be $(k-1,1)$-oriented, although the
running time for computing the orientation may increase with the component size.  The earlier
result of the second author established that in expectation, this
could be done efficiently for $m = (1-\epsilon) \frac{n}{k(k-1)}$ in the case $k=3$.
The connection above complements the earlier result by establishing
feasibility, although not necessarily efficiency, when the number of
edges differs from $\frac{n}{k(k-1)}$ by a sublinear term.

\begin{theorem} (Theorem 10 in \cite{KL-hypergraph-excess}.)  Let $k \geq
  3$ be a fixed integer, and let $m = \frac{n}{k(k-1)} + t(n)$, where
  $t(n)$ is any function of higher order than $n^{2/3}$ but smaller order
  than $n^{2/3} \big( \frac{\log n}{\log \log n} \big)^{1/3}$.  Then with
  high probability, $H_{n,m;k}$ consists of one large complex component and
  some number of small components which are either hypertrees or unicyclic.
\end{theorem} 

\noindent \textbf{Remark.} Clearly, adding more edges only creates more
complex components, so the upper bound on $t(n)$ plays a role only in
limiting the number of large complex components.

\medskip

Therefore, as soon as we exceed $\frac{n}{k(k-1)}$ by even a
sublinear deviation, an obstruction appears, and hence
$(k-1,1)$-orientability fails.  Note that we cannot bound the size
of the complex component, and in fact its size grows with $n$.  There
remains a window of width roughly $n^{2/3}$ between the lower and
upper bounds.  It is worth noting that for the case of graphs, this is also
well-understood, and when $m = \frac{n}{k(k-1)} + cn^{2/3}$ for (positive
or negative) constants $c$, there is a constant probability of having a
complex component.  See, e.g., the discussion in the book
\cite{AlonSpencer}.

\section{High-probability running time bound}

In this section we present and analyze a simple algorithm for finding
$(k-1,k)$-orientations.  The first step is to observe that the running time
to orient each new edge is $O(k^2 s)$, where $s$ is the size of the
connected component formed by the new edge.

\subsection{Algorithm description}\label{sec:algorithm}

The algorithm works by iteratively extending an orientation to more and
more edges. We assume that vertices of each edge $e$ can be traversed using
methods $e$.first() (which returns an arbitrary vertex) and $e$.next$(v)$
(which gives the next node in the order after $v$, cycling back to
$e$.first() when all vertices have been traversed).  For $v\in V$ let
$T[v]$ refer to the edge that is oriented towards $v$, where $T[v]=\bot$ if
no edge is oriented towards~$v$. We maintain an array indexed by $V$ that
initially has all entries set to $\bot$. An edge $e$ is directed to $k-1$
vertices by calling the following procedure, generalizing the procedure
of~\cite{AP11}. We use the notation $\leftrightarrow$ to indicate exchange
of two variable values.

\begin{minipage}{\linewidth}
\begin{tabbing}
  xx\=xx\=xx\=xx\=xx\=\kill
  \+\+\\
  {\bf procedure} {\sc Orient}$(e)$\+\\
    {\bf for} $i:=1$ to $k-1$ {\bf do} \+\\
	$\tau = e$\\
	$v = e$.first()\\
	{\bf while} $\tau\ne\bot$\+\\
	   $v = \tau$.next($v$)\\
       $\tau \leftrightarrow T[v]$\-\\
    {\bf end while}\-\\
  {\bf end for}
\end{tabbing}
\end{minipage}

\medskip

When {\sc Orient} is called, each member of the set of previously oriented
edges $E_1$ appears $k-1$ times in $T$. The procedure runs a while loop
$k-1$ times that (if it terminates) inserts $e$ in $T[v]$ for some $v\in
e$, while ensuring that each edge $e'\in E_1$ is still oriented towards
$k-1$ positions in $T$. The invariant of the while loop is that all edges
in $E_1$ are oriented towards $k-1$ vertices, and $e$ is oriented towards
$i$ vertices, with one exception: If $\tau\ne\bot$ the edge $\tau$ which is
oriented towards one vertex less. Clearly, once $i=k-1$ and $\tau = \bot$
we have oriented all edges in $E_1 \cup \{e\}$.

We claim that the procedure always terminates if an orientation exists, and
more specifically that the time spent if $e$ is in a component of size $s$
is $O(k^2 s)$. (Some stopping criterion is needed for termination in case no
orientation exists, but this is left out for simplicity.) Suppose the while
loop does not stop, i.e., it goes through an infinite sequence of edges.
Let $e_1,e_2,e_3,\dots$ denote this edge sequence, with consecutive
identical edges combined into a single occurrence. We observe that there
can be at most $k-1$ consecutive iterations involving a particular edge.
Notice also that edge $e_i$ shares at least one vertex with edge $e_{i+1}$
for each $i$. Consider a minimal subsequence $e_i,\dots,e_j$ containing 3
such occurrences of some edge, and without loss of generality, assume that
$e=1$. Let $\ell_1$ and $\ell_2$, $1\leq \ell_1<\ell_2<j$, be the indexes
of the edge in this subsequence that first appears for the second time (so
$\ell_2$ is minimal).  Since all previously fully-oriented edges already
are oriented towards all but one of their $k$ vertices, one observes that
then $e_{\ell_{2}+t}=e_{\ell_{1}-t}$ for $t=0,\dots,\ell_{1}-1$. This means
that the $\ell_2-1$ distinct edges $e_1,\dots,e_{\ell_{2}-1}$ contain
exactly $(\ell_2-1)(k-1)$ distinct vertices. From $e_{\ell_2}$ to
$e_{\ell_2+\ell_1-1} = e_1$, the edges encountered are all repeats (in
reverse order) of those already seen.  After $e_{\ell_{2}+\ell_{1}-1}$
(which is equal to $e_1$) each new edge introduces at most $k-1$ new
vertices until we reach an edge that overlaps with a previously visited
edge and only $k-2$ vertices are introduced. At that point we have visited
a set of edges having less than $k-1$ available vertices on average,
meaning that no $(k-1,k)$-orientation exists.

Therefore, the length of the edge sequence is at most $2s$, while the
number of consecutive identical edges consolidated into each element is at
most $k-1$.  Since the while loop is run $k-1$ times for each new edge to
orient, we conclude that the full orientation of the edge completes in $O(k^2
s)$ time.

\subsection{Probabilistic tools}

We will need the following version of the Chernoff bound (see e.g.~\cite{MRbook}):

\begin{theorem}
  \label{thm:chernoff}
  For any $0 < \epsilon < 1$, every binomial random variable $X$ with
  mean $\mu$ satisfies
  \begin{displaymath}
    \pr{X < (1-\epsilon) \mu} < e^{-\frac{\epsilon^2}{2} \mu}
    \quad\quad
    \text{and}
    \quad\quad
    \pr{X > (1+\epsilon) \mu} < e^{-\frac{\epsilon^2}{3} \mu} \,.
  \end{displaymath}
\end{theorem}

A {\em supermartingale\/} is a sequence
$X_0, X_1, \ldots$ of random variables such that each conditional
expectation $\E{X_{t+1} \mid X_0, \ldots, X_t}$ is at most $X_t$.
Azuma's inequality (see e.g.~\cite{MRbook}) states the following:

\begin{theorem}
  \label{thm:azuma}
  Let $X_0, \ldots, X_n$ be a supermartingale, with bounded differences
  $|X_{t+1} - X_t| \leq C$.  Then for any $\lambda \geq 0$,
  \begin{displaymath}
    \pr{X_n \geq X_0 + \lambda}
    \ \leq \
    \exp\left\{-\frac{\lambda^2}{2 C^2 n}\right\}.
  \end{displaymath}
\end{theorem}

\subsection{Analysis of random hypergraphs}

Throughout, we impose explicit bounds that keep $n$ ``sufficiently large''
in order to simplify our calculations.  Recall that $H_{n,m;k}$ is the
random $k$-uniform hypergraph obtained by uniformly sampling one with
$n$-vertices and $m$ hyperedges.  In this section, it will be substantially
more convenient for us to work with a process that exhibits more
independence.  Specifically, we consider instead the following sequential
process, which fortunately is quite similar to the original $H_{n,m;k}$.

\begin{lemma}
  \label{lem:process-Hnm}
  Let $n > k \geq 2$, with $n > 2000$.  Consider the random
  hypergraph process $H_0, H_1, \ldots$, where $H_0$ is the empty
  hypergraph with $n$ isolated vertices.  At each time $t$, sample $k$
  vertices independently and uniformly at random.  If they are
  distinct, and form a hyperedge which does not yet appear in $H_t$,
  then add it to form $H_{t+1}$.  Otherwise, let $H_{t+1} = H_t$.
  Then, with probability at least $1 - n^{-1}$, in the first
  $\frac{n}{k(k-1)}$ rounds, the number of times that we do not add is
  at most $\log n$.
\end{lemma}

\noindent \textbf{Proof.}\, At time $t+1$, a union bound shows that
the probability that the $k$ sampled vertices are not distinct is at
most
\begin{displaymath}
  \frac{1}{n} + \frac{2}{n} + \cdots \frac{k-1}{n}
  = \frac{k(k-1)}{2n} \,.
\end{displaymath}
Also, since $t < \frac{n}{k(k-1)}$, the probability that the $k$
sampled vertices form a previously-added hyperedge is
\begin{displaymath}
  \frac{t k!}{n^k}
  <
  \frac{(k-2)!}{n^{k-1}} < \frac{1}{n} \,,
\end{displaymath}
where we used $n > k$ for the final bound.  Thus the probability that
$H_{t+1} = H_t$ is at most $\frac{k(k-1)}{n}$, and so the probability
that this happens at least $s = \log n$ times in the first
$\frac{n}{k(k-1)}$ rounds is at most
\begin{displaymath}
  \binom{ \frac{n}{k(k-1)} }{ s }
  \left(
    \frac{k(k-1)}{n}
  \right)^{s}
  \leq
  \left(
    \frac{e \cdot \frac{n}{k(k-1)}}{s}
  \right)^{s}
  \left(
    \frac{k(k-1)}{n}
  \right)^{s} 
  =
  \left(
    \frac{e}{s}
  \right)^{s} 
  =
  \left(
    \frac{e}{\log n}
  \right)^{\log n} \,,
\end{displaymath}
which is below $n^{-1}$ for all $n > e^{e^2}$.  \hfill $\Box$

\bigskip

It is sometimes more convenient to work with the related model
$H_{n,p;k}$, which is the random $k$-uniform hypergraph formed by
taking each of the $\binom{n}{k}$ potential hyperedges independently
with probability $p$.  Fortunately, the behavior of $H_{n,p;k}$
closely approximates that of $H_{n,m;k}$.

\begin{lemma}
  \label{lem:coupling}
  Assume that $0 < \epsilon < \frac{1}{2}$, $k \geq 2$, and
  $\frac{n}{\log n} > \frac{100 k^2}{\epsilon^2}$.  Let $m = (1 -
  \epsilon) \frac{n}{k(k-1)}$, and $p = (1 - 0.8 \epsilon)
  \frac{(k-2)!}{n^{k-1}}$.  Then we may couple the probability spaces
  such that $H_{n,m;k}$ is contained in $H_{n,p;k}$ with probability
  at least $1 - n^{-1}$.
\end{lemma}

\noindent \textbf{Proof.}\, Couple the probability spaces with the
random hypergraph process, which corresponds to a uniformly random
permutation of all $\binom{n}{k}$ potential edges.  Then, $H_{n,m;k}$
corresponds to forming a hypergraph with precisely the first $m$
edges in this permutation, and $H_{n,p;k}$ corresponds to generating
an independent random variable $X \sim \bin{\binom{n}{k}, p}$, and
then taking the first $X$ edges in the permutation.

Therefore, it suffices to show that $X \geq (1-\epsilon)
\frac{n}{k(k-1)}$ with high probability.  We calculate
\begin{displaymath}
  \E{X} 
  = \binom{n}{k} p
  \geq ( 1 - 0.8 \epsilon ) 
  \frac{(n-k)^k}{k(k-1) n^{k-1}} \,.
\end{displaymath}
Next, observe that if $\big( \frac{n-k}{n} \big)^k \geq 1 - 0.01
\epsilon$, then we will have $\E{X} \geq ( 1 - 0.81 \epsilon )
\frac{n}{k(k-1)}$.  The Chernoff bound (Theorem \ref{thm:chernoff})
would then give
\begin{displaymath}
  \pr{X < (1 - \epsilon) \frac{n}{k(k-1)}}
  <
  e^{-\frac{(0.19 \epsilon)^2}{2} ( 1 - 0.81 \epsilon ) \frac{n}{k(k-1)} }  \,.
\end{displaymath}
Using $\epsilon < \frac{1}{2}$, and $n > \frac{100 k^2}{\epsilon^2}
\log n$, we conclude that this probability is at most $n^{-1.07}$.  It
remains to show that $\big( \frac{n-k}{n} \big)^k \geq 1 -
\frac{\epsilon}{100}$.  After rearrangement, we see that this is
equivalent to
\begin{equation}
  \label{ineq:coupling-need}
  n \geq \frac{k}{1 - \left( 1 - \frac{\epsilon}{100} \right)^{1/k}} \,.
\end{equation}
However, $e^{-x} \leq 1 - \frac{x}{2}$ for all $0 \leq x \leq 1$, so
\begin{displaymath}
  \left( 1 - \frac{\epsilon}{100} \right)^{1/k}
  \leq
  e^{-\frac{\epsilon}{100k}}
  \leq
  1 - \frac{\epsilon}{200k} \,.
\end{displaymath}
This, together with our assumption that $n > \frac{100
  k^2}{\epsilon^2} \log n > \frac{200 k^2}{\epsilon}$, produces
\eqref{ineq:coupling-need}.  \hfill $\Box$

\vspace{3mm}

\begin{lemma}
  \label{lem:log-components}
  Let $0 < \epsilon < \frac{1}{2}$ and $n > \frac{200
    k^2}{\epsilon}$.  Let $p = (1 - 0.8 \epsilon)
  \frac{(k-2)!}{n^{k-1}}$.  In the random hypergraph $H_{n,p;k}$, with
  probability at least $1 - n^{-1}$, all connected components are of
  size at most $\frac{16k}{\epsilon^2} \log n$.
\end{lemma}

\noindent \textbf{Proof.}\, Let $V$ be the vertex set of the entire
hypergraph.  Let $v$ be a fixed vertex, and let the random variable
$X_v$ be the size of the connected component containing $v$.  We
generate $X_v$ by exposing hyperedges one at a time via
breadth-first-search.  Specifically, we maintain time-varying sets
$A_t$ of distinct active vertices and $B_t$ of completed vertices, and
build a labeling of the vertices $v_1, v_2, \ldots, v_n$, initializing
$A_0 = \{v\}$ and $B_0 = \emptyset$.  At time $t$, we arbitrarily
select a vertex $w \in A_t$ (if $A_t$ is empty, we stop), define the
label $v_t = w$, and set $A_{t+1} = A_t \setminus \{w\}$ and $B_{t+1}
= B_t \cup \{w\}$.  Also, we expose all hyperedges which have exactly
$k-1$ vertices in $V \setminus \{v_1, \ldots, v_{t-1}\}$, together
with $w$ as the $k$-th vertex.  Here, ``expose'' means that we reveal
whether or not the potential hyperedge in fact appears in this
particular realization of $H_{n,p;k}$.  Finally, for each vertex other
than $w$ which is in a newly exposed hyperedge, we add a single copy
of that vertex to $A_{t+1}$.

Importantly, we never expose the same hyperedge twice, because the
hyperedges exposed at time $t$ have the property that their smallest
labeled vertex is precisely $v_t$.  Therefore, the decisions are
independent at each stage, and the number of vertices added to the
queue is stochastically dominated by $(k-1)$ times the Binomial random
variable $\bin{\binom{n}{k-1}, p}$.  In particular, if we define the
random variables $Y_t = |A_t|$, then each successive difference
$Y_{t+1} - Y_t$ is stochastically dominated by $(k-1)
\bin{\binom{n}{k-1}, p} - 1$.  Therefore, if we define the infinite
sequence $Z_t$ as $Z_0 = 1$, $Z_{t+1} = Z_t + (k-1)
\bin{\binom{n}{k-1}, p} - 1$, we may couple the probability spaces
such that $Y_t \leq Z_t$ until $Y_t$ hits 0 (the breadth-first-search
is exhausted).

Let $T = \frac{16k}{\epsilon^2} \log n$.  Since a binomial random
variable is the sum of independent Bernoullis, the sum of independent
and identically distributed binomials is another binomial.  Thus the
distribution of $Z_T$ is precisely $1 + (k-1) \bin{\binom{n}{k-1} T,
  p} - T$.  The Chernoff bound will control the probability that $Z_T
\geq 1$, and this will be sufficient because if the integer $Z_T < 1$,
then the breadth-first-search must have completed, as $Y_t \leq Z_t$
during it.  Observe that $Z_T \geq 1$ happens precisely when
$\bin{\binom{n}{k-1} T, p} \geq \frac{T}{k-1}$.  Yet the expectation
of this binomial is
\begin{displaymath}
  \binom{n}{k-1} T (1 - 0.8\epsilon) \frac{(k-2)!}{n^{k-1}}
  \leq
  (1-0.8\epsilon) \frac{T}{k-1} \,,
\end{displaymath}
so when $Z_T \geq 1$, that binomial exceeds its expectation $\mu$ by a
factor of at least $0.8\epsilon$.  Hence
\begin{displaymath}
  \pr{Z_T \geq 1}
  \leq
  e^{-\frac{(0.8 \epsilon)^2}{3} \mu} \,.
\end{displaymath}
To continue, we need a lower bound on $\mu$.  At the end of the proof
of the previous lemma, we showed that $n \geq \frac{200k^2}{\epsilon}$
implies that $\big( \frac{n-k}{n} \big)^k \geq 1 - 0.01 \epsilon$.
Since $\big( \frac{n-k}{n} \big)^{k-1} > \big( \frac{n-k}{n}
\big)^k$, and we assume $\epsilon < \frac{1}{2}$, we therefore have
that
\begin{displaymath}
  \mu 
  = 
  \binom{n}{k-1} T (1 - \epsilon) \frac{(k-2)!}{n^{k-1}}
  \geq
  ( 1 - 0.81 \epsilon ) \frac{T}{k-1} 
  \geq
  0.595 \cdot \frac{T}{k-1} \,.
\end{displaymath}
Thus using $T = \frac{16k}{\epsilon^2} \log n$, we have
\begin{displaymath}
  \pr{Z_T \geq 1}
  <
  e^{-\frac{(0.8 \epsilon)^2}{3} \cdot 0.595 \cdot \frac{T}{k-1}} 
  <
  n^{-2} \,,
\end{displaymath}
i.e., a fixed vertex $v$ has probability at least $1 - n^{-2}$ of
having its component size at most $\frac{16k}{\epsilon^2} \log n$.  A
final union bound over the $n$ vertices yields the desired result.
\hfill $\Box$

\medskip

We now move to introduce the key parameter which characterizes the overall
running time of our algorithm.  This parameter has been successfully used
to analyze various discrete random processes, ranging from percolation
(where its name originated from statistical physics) to the theory of
random graphs and stochastic coalescence processes.

\begin{definition}
  \label{def:susceptibility}
  The \textbf{susceptibility} of a given hypergraph $H$, denoted by
  $\chi(H)$, is the expected size of the component which contains a
  uniformly random vertex.  Equivalently, if the connected components
  are $C_1, C_2, \ldots, C_t$, then $\chi = \frac{1}{n} \sum_i
  |C_i|^2$.
\end{definition}

It turns out that the susceptibility evolves smoothly with the number of
edges $m$.  The following theorem applies the Differential Equations method
to estimate its growth.  Its analysis builds upon the approach used in
\cite{BFKLS}, but improves the error bounds from exponential to polynomial
(in both $\frac{1}{\epsilon}$ and $n$).

\begin{theorem}
  \label{thm:susceptibility}
  Let $0 < \epsilon < \frac{1}{2}$ be given, and assume that
  $\frac{n}{\log^6 n} > \frac{40000 k^6}{\epsilon^{12}}$.  Let $m =
  (1-\epsilon) \frac{n}{k(k-1)}$.  With probability at least $1 - 3
  n^{-1}$, the random $k$-uniform hypergraph $H_{n,m;k}$ has
  susceptibility at most
  \begin{equation}
    \label{eq:suscep:target}
    \frac{1}{\epsilon} 
    + \frac{ 200 k^3 \log^3 n }{\epsilon^7 \sqrt{n}} \,.
  \end{equation}
\end{theorem}

\noindent \textbf{Proof.}\, Define
\begin{displaymath}
  T = (1-\epsilon) \frac{n}{k(k-1)} \,.
\end{displaymath}
Consider the random hypergraph process $H_0, H_1, \ldots$ of Lemma
\ref{lem:process-Hnm}.  We will run this process to time $T + \log n$
which by Lemma \ref{lem:process-Hnm} will contain $H_{n,m;k}$ with
probability at least $1 - n^{-1}$, because $\log n < \epsilon \cdot
\frac{n}{k(k-1)}$.  It therefore suffices to show that with
probability at least $1 - 2n^{-1}$, the susceptibility of $H_{T+\log
  n}$ is at most \eqref{eq:suscep:target}.  We track the evolution of
susceptibility by defining $X_t$ to be the susceptibility of $H_t$.
Suppose that in the $(t+1)$-st round, the $k$ vertices of the incoming
hyperedge lie in components $C_1, \ldots, C_k$, where some of the
components may be repeated.  Then, the susceptibility increases by at
most
\begin{displaymath}
  \frac{1}{n} \left[
    (|C_1| + \cdots + |C_k|)^2 - 
    (|C_1|^2 + \cdots + |C_k|^2)
  \right]
  =
  \frac{2}{n} \sum_{r<s} |C_r| |C_s| \,,
\end{displaymath}
with equality only if $C_1, \ldots, C_k$ are distinct.  Define the
filtration $\mathcal{F}_0, \mathcal{F}_1, \ldots$ such that
$\mathcal{F}_t$ captures the outcomes up to and including time $t$.
Let us bound $\E{X_{t+1} - X_t \mid \mathcal{F}_t}$.  For this, let
$c_1, \ldots, c_z$ be the component sizes after time $t$.  Since our
process selects $k$ independent vertices for the next hyperedge, and
the hyperedge is added only if they are distinct from each other, and
form a new hyperedge, we then have
\begin{align*}
  \E{X_{t+1} - X_t \mid \mathcal{F}_t}
  &\leq
  \sum_{i_1, \ldots, i_k \in [z]}
  \left( \frac{c_{i_1}}{n} \cdot \frac{c_{i_2}}{n} \cdots \frac{c_{i_k}}{n} \right) 
  \cdot \frac{2}{n} \sum_{1 \leq r < s \leq k} c_{i_r} c_{i_s} \\
  &=
  \frac{2}{n} \cdot \binom{k}{2} 
  \cdot
  \sum_{i_1, \ldots, i_k \in [z]}
  \frac{c_{i_1}^2 c_{i_2}^2 c_{i_3} c_{i_4} \cdots c_{i_k}}{n^k} \\
  &=
  \frac{k(k-1)}{n} \left( 
    \sum_{i_1} \frac{c_{i_1}^2}{n}
  \right)
  \left( 
    \sum_{i_2} \frac{c_{i_2}^2}{n}
  \right)
  \left( 
    \sum_{i_3} \frac{c_{i_3}}{n}
  \right)
  \cdots
    \left( 
    \sum_{i_k} \frac{c_{i_k}}{n}
  \right) \\
  &=
  \frac{k(k-1)}{n} \left( X_t \right) \left( X_t \right) (1) \cdots (1)
  = \frac{k(k-1)}{n} X_t^2 \,.
\end{align*}
This suggests that the evolution of $X_t$ may resemble that of the
differential equation $x'(\theta) = k(k-1) x(\theta)^2$, where we
parameterize $\theta = \frac{t}{n}$.  So, let us define 
\begin{displaymath}
  x(\theta) = \frac{1}{1 - k(k-1) \theta} \,,
\end{displaymath}
which is the solution of that differential equation with initial
condition $x(0) = 1$.  We use $x(\theta)$ to convert $(X_t)$ into a
supermartingale, defining
\begin{displaymath}
  Y_t = X_t - x\left( \frac{t}{n} \right) - f\left( \frac{t}{n} \right) \Delta \,,
\end{displaymath}
where
\begin{displaymath}
  f(\theta) = \frac{1}{ ( 1 - k(k-1) \theta )^3} \,,
  \quad\quad
  \Delta = \frac{199 k^3 \log^3 n}{\epsilon^4 \sqrt{n}}
\end{displaymath}
is the solution to the differential equation
\begin{displaymath}
  f'(\theta) = 3 k(k-1) \cdot x(\theta) f(\theta) \,;
  \quad\quad
  f(0) = 1 \,.
\end{displaymath}
Also, define $E_t$ to be the event that \textbf{(i)} $X_t \leq x\big(
\frac{t}{n} \big) + f\big( \frac{t}{n} \big) \Delta$ and \textbf{(ii)}
all components of $H_t$ have size at most $\frac{16k}{\epsilon^2}
\log n$.  Then, define the stopping time $\tau$ to be either $T$, or
the first moment that $E_t$ fails.  Define
\begin{displaymath}
  Z_t = Y_{\min\{t, \tau\}} \,.
\end{displaymath}
Then, 
\begin{displaymath}
  \E{Z_{t+1} - Z_t \mid \mathcal{F}_t, \overline{E_t}} = 0 \,
\end{displaymath}
and
\begin{displaymath}
  \E{Z_{t+1} - Z_t \mid \mathcal{F}_t, E_t}
  \leq
  \frac{k(k-1)}{n} X_t^2 
  - \left[
    x\left( \frac{t+1}{n} \right) - x\left( \frac{t}{n} \right)
  \right]
  -
  \left[
    f\left( \frac{t+1}{n} \right) - f\left( \frac{t}{n} \right)
  \right]
  \Delta \,,
\end{displaymath}
which by convexity of $x(\theta)$ and $f(\theta)$ is at most
\begin{align*}
  \E{Z_{t+1} - Z_t \mid \mathcal{F}_t, E_t}
  &\leq
  \frac{k(k-1)}{n} X_t^2 
  - \frac{1}{n} x'\left( \frac{t}{n} \right)
  - \frac{1}{n} f'\left( \frac{t}{n} \right)
  \Delta \\
  &\leq
  \frac{k(k-1)}{n} \left[
    x\left( \frac{t}{n} \right) + f\left( \frac{t}{n} \right) \Delta
  \right]^2
  - \frac{1}{n} x'\left( \frac{t}{n} \right)
  - \frac{1}{n} f'\left( \frac{t}{n} \right)
  \Delta \\
  &=
  \frac{k(k-1)}{n} \left[
    2 x\left( \frac{t}{n} \right) f\left( \frac{t}{n} \right) \Delta
    +
    f\left( \frac{t}{n} \right)^2 \Delta^2
  \right]
  - \frac{1}{n} f'\left( \frac{t}{n} \right) \Delta \,.
\end{align*}
Our condition on $n$ is essentially equivalent to the fact that $\Delta \leq
\epsilon^2$.  This implies that over the range $0 \leq \theta \leq
\frac{1-\epsilon}{k(k-1)}$, we always have
\begin{equation}
  \label{ineq:need}
  f(\theta) \Delta \leq x(\theta)  \,,
\end{equation}
so
\begin{displaymath}
  \E{Z_{t+1} - Z_t \mid \mathcal{F}_t, E_t}
  \leq
  \frac{\Delta}{n} \left[
    3k(k-1) \cdot x\left( \frac{t}{n} \right) f\left( \frac{t}{n} \right) 
    - f'\left( \frac{t}{n} \right)
  \right]
  =
  0 \,,
\end{displaymath}
and we conclude that $Z_t$ is in fact a supermartingale.  Furthermore,
by part (ii) of the definition of $E_t$, we know that the addition of a
single hyperedge cannot increase the susceptibility by more than
\begin{displaymath}
  \frac{1}{n} \left[
    \left( k \cdot \frac{16k}{\epsilon^2} \log n \right)^2
    -
    k \cdot \left( \frac{16k}{\epsilon^2} \log n \right)^2
  \right]
  <
  \frac{256 k^4 \log^2 n}{\epsilon^4 n} \,.
\end{displaymath}
Since $x(\theta)$ and $f(\theta)$ are both increasing functions, this
is an upper bound for the incremental change $Z_{t+1} - Z_t$.  On the
other hand, the susceptibility can never decrease, and on the range
$\theta < \frac{1-\epsilon}{k(k-1)}$, the derivatives $x'(\theta)$ and
$f'(\theta)$ increase to $\frac{k(k-1)}{\epsilon^2}$ and $\frac{3
  k(k-1)}{\epsilon^{4}}$, respectively.

Since $x(\theta)$ and $f(\theta)$ are convex, we conclude that as $t$
ranges from 0 to $T$, the maximum one-step change in $Z_t$ is bounded
by
\begin{displaymath}
  C = \max\left\{
    \frac{256 k^4 \log^2 n}{\epsilon^4 n} \,,
    \frac{k(k-1)}{\epsilon^2} \cdot \frac{1}{n} + \frac{3k(k-1)}{\epsilon^4} \cdot \frac{\Delta}{n}
  \right\} 
  =
  \frac{256 k^4 \log^2 n}{\epsilon^4 n}
  \,.
\end{displaymath}
Yet $Z_0 = -\Delta$, so Hoeffding-Azuma (Theorem \ref{thm:azuma}) implies that
\begin{displaymath}
  \pr{ Z_T \geq 0 }
  \leq
  \exp\left\{
    -\frac{\Delta^2}{2 C^2 T}
  \right\} 
  <
  \exp\left\{
    -\frac{4 n \log^2 n}{k^2 T}
  \right\} 
  <
  n^{-1} \,.
\end{displaymath}
Also, by Lemma \ref{lem:log-components}, the probability that $H_T$
has a component with size exceeding $\frac{16k}{\epsilon^2} \log n$ is
at most $n^{-1}$.  Hence with probability at least $1 - 2n^{-1}$, we
have that both $Z_T < 0$ and all components of $H_T$ have size at
most $\frac{16k}{\epsilon^2} \log n$.  When this happens, we must
never have had any $E_t$ fail, and hence we conclude that $Y_T = Z_T <
0$, and so the susceptibility after $T$ rounds is
\begin{displaymath}
  X_T 
  \leq 
  x\left(\frac{T}{n}\right) + f\left(\frac{T}{n}\right) \Delta
  =
  x\left(\frac{1-\epsilon}{k(k-1)}\right) + f\left(\frac{1-\epsilon}{k(k-1)}\right) \Delta
  =
  \frac{1}{\epsilon} + \frac{1}{\epsilon^3} \cdot \frac{199 k^3 \log^3 n}{\epsilon^4 \sqrt{n}} \,.
\end{displaymath}

Adding $\log n$ more rounds to reach time $T + \log n$, we see that
these can link at most $k \log n$ clusters, and since we conditioned
on all clusters having size at most $\frac{16k}{\epsilon^2} \log n$,
this can further increase the susceptibility by at most
\begin{displaymath}
  \frac{1}{n} \cdot \left( k \log n
    \cdot \frac{16k}{\epsilon^2} \log n \right)^2
  =
  \frac{256 k^4 \log^4 n}{\epsilon^4 n}
  <
  \frac{k^3 \log^3 n}{\epsilon^7 \sqrt{n}} \,,
\end{displaymath}
by our initial assumption on the size of $n$.  Therefore, with
probability at least $1 - 2 n^{-1}$, the total susceptibility after $T
+ \log n$ rounds is at most $\frac{1}{\epsilon} + \frac{200 k^3 \log^3
  n}{\epsilon^7 \sqrt{n}}$, as required.  \hfill $\Box$

\medskip

We now combine all of our results to produce our main theorem, which
provides a single high-probability bound for the final sum of squared
component sizes in $H_{n,m;k}$.

\medskip

\noindent \textbf{Proof of Theorem \ref{thm:main}.}\, As explained in
section~\ref{sec:algorithm}, the time for inserting a key in a component of
size $s$ is $O(k^2 s)$.  This means that if the final hypergraph contains a
component of size $s$, it took only $O(\sum_{i=1}^s k^2 s)$ time to insert
all edges of that component, i.e., $O(k^2 s^2)$ operations. Each edge is in
exactly one component, and we recognize that summing the squares of the
final component sizes gives exactly $n$ times the final susceptibility.
Thus, we can bound the total running time by $O(k^2 n)$ times the final
susceptibility, which by Theorem \ref{thm:susceptibility} is bounded by a
constant with high probability. \hfill $\Box$

\section{Acknowledgments}

We thank Alan Frieze for invigorating discussions which inspired us to
pursue this project.

\bibliographystyle{abbrv}
\bibliography{threshold23}

\end{document}